\documentclass[preprint,12pt]{article}

\usepackage{xcolor}
\usepackage{enumerate}
\usepackage{subfiles}
\usepackage{changepage}
\usepackage[font={small,it}]{caption}
\usepackage{verbatim}
\usepackage[english]{babel}
\usepackage[utf8]{inputenc}
\usepackage[margin=2cm]{geometry}
\usepackage{float}
\usepackage{times}
\usepackage{soul}
\usepackage{url}
\usepackage[hidelinks]{hyperref}
\usepackage[utf8]{inputenc}
\usepackage{graphicx}
\usepackage{amsmath}
\usepackage{amsfonts}
\usepackage{amsthm}
\usepackage{booktabs}
\usepackage{subcaption}

\theoremstyle{definition}
\newtheorem{definition}{Definition}
\newtheorem{theorem}{Theorem}
\newtheorem{proposition}{Proposition}
\newtheorem{lemma}{Lemma}
\newtheorem{corollary}{Corollary}
\newtheorem{assumption}{Assumption}
\theoremstyle{remark}
\newtheorem*{remark}{Remark}

\newcommand{\agentset}{\mathcal{N}}
\newcommand{\edgeset}{\mathcal{E}}
\newcommand{\actionset}[1]{S_{#1}}

\newcommand{\NE}{\mathbf{\bar{x}}}

\newcommand{\A}{\mathbf{A}}
\newcommand{\B}{\mathbf{B}}

\newcommand{\x}{\mathbf{x}}
\newcommand{\y}{\mathbf{y}}

\newcommand{\w}{\mathbf{w}}

\newcommand{\R}{\mathbb{R}}

\usepackage{listings}

\usepackage{multirow}

\usepackage{fancyhdr}

\title{Asymptotic Convergence and Performance of Multi-Agent Q-learning Dynamics}
\author{
  Aamal Abbas Hussain\\
    \texttt{aamal.hussain15@imperial.ac.uk}
  \and
  Francesco Belardinelli\\
  \texttt{francesco.belardinelli@imperial.ac.uk}
  \and 
  Georgios Piliouras\\
  \texttt{georgios@sutd.edu.sg}
}
\date{}
\begin{document}
    \maketitle

    \begin{abstract}
Achieving convergence of multiple learning agents in general $N$-player games
is imperative for the development of safe and reliable machine learning (ML) algorithms and their application to
autonomous systems. Yet it is known that, outside the bounds of simple two-player
games, convergence cannot be taken for granted.

To make progress in resolving this problem, we study the dynamics of smooth Q-Learning, a popular reinforcement learning algorithm which quantifies the tendency for learning agents to explore their state space or exploit their payoffs. We show a sufficient condition on the rate of exploration such that the Q-Learning dynamics is guaranteed to converge to a unique equilibrium in any game. We connect this result to games for which Q-Learning is known to converge with arbitrary exploration rates, including weighted Potential games and weighted zero sum polymatrix games.

Finally, we examine the performance of the Q-Learning dynamic as measured by the Time Averaged Social Welfare, and comparing this with the Social Welfare achieved by the equilibrium. We provide a sufficient condition whereby the Q-Learning dynamic will outperform the equilibrium even if the dynamics do not converge.
    \end{abstract}

    \section{Introduction}

    Understanding the behaviour of multi-agent learning systems has been a hallmark problem of game theory and online
    learning. The requirement is that agents must explore potentially sub-optimal decisions whilst interacting with
    other agents to ultimately maximise their long-term reward. In contrast to online learning with a single agent,
    this poses a fundamentally non-stationary problem, in which convergence to an equilibrium is not always
    guaranteed.

    In fact, recent work has consistently found that, when learning on games, agents may present a wide array of
    behaviours. This includes cycles \cite{piliouras:cycles,piliouras:poincare,galla:cycles}, and even chaos \cite
    {sato:qlearning,sato:rps,galla:complex,sanders:chaos,piliouras:arbitrarilycomplex}. Furthermore, the equilibria of
    a game need not be unique, so even if convergence is guaranteed, it may be to one of many (or even a continuum) of
    equilibria.  Thus, predicting the behaviour of online learning in games with many players becomes a particularly
    challenging problem.
    
    Yet it remains an important problem to solve. Recent advances in machine learning require training multiple neural
    networks for applications to generative models \cite{che:gan,hoang:mgan}. In order for such applications to be
    realised, it is required that that training provably converges to an equilibrium. Of equal importance is that this
    equilibrium be unique, so that the outcome remains consistent regardless of initial conditions. 
    A unique equilibrium guarantees not only the reproducibility of the system, but also ensures that desired behaviours will persist, even if the system is perturbed from its desired state. 
    Similarly, the most complex tasks often require the interaction of multiple autonomous agents \cite{hamann:swarm}. This again
    requires that agents are able to reliably equilibrate their behaviour.
    
    There is a strong, and ongoing, effort in the research community to understand these learning behaviours, with
    positive convergence results being found in an assortment of game structures. For instance, games with two players
    and two actions are well understood \cite{pangallo:taxonomy,metrick:fp,galstyan:2x2}. Beyond this, some of the most
    widely studied games are potential games, in which agents collaborate to maximise a shared global function, and
    zero sum games (and its network variants), in which agents are in competition. Indeed, it has been found that a
    number of learning algorithms, including Fictitious Play \cite{shoham:mas}, Q-Learning \cite
    {sutton:barto,tuyls:qlearning,sato:qlearning},  Replicator Dynamics \cite
    {smith:replicator,hofbauer:egd,hofbauer:book} all converge to equilibria (though not always unique) in potential
    games \cite{harris:fp,piliouras:potential}. In zero sum games, the former two converge to a unique fixed
    point \cite{ewerhart:fp,piliouras:zerosum}, whereas the latter is known to cycle, always maintaining its distance
    from the equilibrium \cite{piliouras:cycles}.
    
    Outside of this class of games, however, the story becomes much more complicated.  A wide array of results show
    learning algorithms may be chaotic in even the simplest games \cite{sato:rps}. Such complex behaviours become even
    more pronounced as the number of agents in the game increases \cite{sanders:chaos}. However, they are also
    influenced by the structure of the game and the parameters of the learning algorithm \cite{pangallo:bestreply}. This dichotomy between the range of possible learning behaviours and the convergence
    requirement of the applications motivate our central question:
\begin{center}
        \textit{Are there learning dynamics such that convergence to a unique equilibrium is guaranteed in any game?}
\end{center}

    \paragraph*{Main Contribution.} To answer this, we study the \emph{(smooth) Q-Learning} dynamics, a popular online
     learning model which captures the behaviour of agents who must balance their tendency to explore their strategy
     space against exploiting their payoffs. 

    In the first of our contributions, we answer the above question in a positive manner. Namely, we show that, through
    sufficient exploration, agents engaged in any game can reach a unique equilibrium. We parameterise the amount of
    exploration required in terms of the size of the game and the number of players. We then revisit previously
    established results of smooth Q-Learning and show that convergence of the algorithm in weighted potential and weighted
    network zero sum games both follow as special cases of our main result. 

   To qualify our convergence result, we consider the \emph{payoff performance} of learning dynamics in terms of the sum
   of payoffs received by all agents. We provide a condition whereby learning dynamics that do not converge can
   outperform the equilibrium. In such a case, it may be beneficial to assume no exploration on the parts of the
   agents. 
   
   This result is supported by our experiments in which examples of games are considered where payoff
   performance degrades as agents are pushed towards an equilibrium through exploration.
    
    \paragraph{Related Work.} Our work applies the framework of monotone games, which encompass a number of classes of
     games, including potential \cite{parise:network} and network zero sum games \cite{kadan:exponential}. Indeed
     recent work has considered the question of designing games which are monotone \cite{tatarenko:monotone}. Most prominent in this is
     the study of network games \cite{melo:network}. These works relate the monotonicity of the game with properties of
     the network. From a design perspective, this makes the study of monotone games rather attractive, as it becomes
     possible to make any game, be it co-operative or competitive, monotone. This fact is also exploited in \cite
     {melo:qre}, in which it was shown how any game, regardless of its structure, can be made monotone by appropriate
     parameter tuning.
    
    In addition, monotone games have been used to design online learning algorithms that 
    converge to a Nash Equilibrium \cite{facchinei:VI, mertikopoulos:concave}. However, many of these require the
    gradient (or estimates) of their cost functions at each step. Ideally, an online learning algorithm should require
    only obtained rewards at each time step. To resolve this, \cite{tatarenko:monotone} derive a distributed algorithm
    which converges to a Nash Equilibrium in monotone games. These results are advanced in \cite
    {tatarenko:nonstrict} in which an algorithm is developed which converges, whilst also achieving no-regret; this was
    the first instance of a no-regret algorithm who could provably converge in a monotone game. In a similar
    manner, \cite{parise:network} show the convergence of the Best Response Algorithm in a class of network games which
    satisfy the monotonicity property. 

    Our work departs from the above by considering Q-Learning, and by lifting strong technical assumptions on the payoff
    functions. Specifically, we do not assume a form of the payoffs as in \cite{parise:network}, or the growth of the
    function as in \cite{tatarenko:nonstrict}. In addition, we require no knowledge of the cost gradients as in \cite
    {mertikopoulos:robust}. Finally, our work considers the generalised class of \emph{weighted} monotone games, rather
    than the unweighted case considered by the above. This class of games is also considered in \cite
    {mertikopoulos:power-management,mertikopoulos:gradient-free}, in which variations of online gradient descent are
    analyzed. However, the former requires weighted strong monotonicity (which is much more restrictive even than
    strict monotonicity) and the latter requires strong assumptions on the parameters of the learning algorithm.
    
    Our work also touches upon the Follow the Regularised Leader (FTRL) dynamic. The strongest convergence result
    regarding this dynamic in continuous time is a negative one: FTRL does not converge to a Mixed Nash
    Equilibrium \cite{flokas:donotmix}, or to any equilibrium in zero sum games \cite{piliouras:cycles}. The strongest
    positive result, and the one most similar to our own is \cite{mertikopoulos:finite} in which it is shown that FTRL
    converges in unweighted strictly monotone games, a more restrictive class that that analysed here. Other positive
    results regard FTRLs local convergence to a strict Nash Equilibrium \cite{mertikopoulos:reinforcement} and its
    convergence in time average \cite{mertikopoulos:finite}. 
    
    Our paper is structured as follows. We begin in Section \ref{sec::Prelims} by outlining the setting and tools through which we analyse convergence of learning. Section \ref{sec::Results} proceeds with our main results, namely that convergence of Q-Learning dynamics can be achieved through sufficient exploration (Theorem \ref{thm::StabilityUnique}), and that non-convergent learning dynamics can outperform the equilibrium (Theorem \ref{thm::PayoffPerformance}). Our experiments in Section \ref{sec::Experiments} validate these results by showing examples of games in which convergence occurs through increased exploration, although at the cost of decreasing payoff across all agents.

  \section{Preliminaries} \label{sec::Prelims}

    In this section, we expand on the necessary background for our main results; specifically how the game model is set
    up, the learning dynamics of interest, and the techniques used in our analysis.

   \subsection{Game Model} \label{sec::GameModel}

    In our study, we consider a game $\Gamma = (\agentset, (\actionset{k}, u_k)_{k \in \agentset})$, where $\agentset$
    denotes a finite set of agents indexed by $k = 1, \ldots, N$. Each agent $k \in \agentset$ is equipped with a
    finite set of actions denoted by $\actionset{k}$ with the number of actions $n_k := |\actionset{k}|$ as well as a payoff function $u_k$. We denote a \emph{mixed strategy} (hereafter just \emph{strategy}) $\x_k$
     of an agent $k$ as a probability vector over its actions. Then, the set of all strategies of agent $k$ is
     $\Delta_k := \left\{ \x_k \in \R^{n_k} \, : \, \sum_i x_{ki} = 1, x_{ki} \geq 0 \right\}$ on which act the payoff
     functions $u_k \, : \Delta_k \times \Delta_{-k} \rightarrow \R$. We denote by $\x := (\x_k)_
     {k \in \agentset} \in \Delta = \times_k \Delta_k$ the \emph{joint strategy} of all agents and, for any $k$, $\x_
     {-k} := (\x_l)_{l \in \agentset \backslash \{k\}}  \in \Delta_{-k}$ the joint strategy of all agents other than
     $k$. 

    For any $\x \in \Delta$, we define the reward to agent $k$ for playing action $i \in \actionset{k}$ as $r_{ki}
    (\x) := \frac{\partial u_{ki}(\x)}{\partial x_{ki}}$. We write $r_k(\x) = (r_{ki}(\x))_{i \in \actionset{k}}$ as
    the concatenation of all rewards to agent $k$. Using this notation we can write $u_k
    (\x) = \langle \x_k, r_k(\x) \rangle$ in which $\langle \x, \y \rangle = \x^\top \y$ denotes the standard inner
    product in $\R^n$. With this in mind, we define the \emph{Equilibrium} of a game.
    \begin{definition}[Equilibrium] \label{def::NE} A joint mixed strategy $\NE \in \Delta$ is an \emph{Equilibrium} if,
     for all agents $k$ and all $\x_k \in \Delta_k$
        \begin{equation} \label{eqn::NE}
            \langle \x_k, r_k(\NE) \rangle \leq \langle \NE_k, r_k(\NE) \rangle
        \end{equation}
    \end{definition} $\NE$ is a \emph{strict equilibrium} if the inequality (\ref{eqn::NE}) is strict for all
     $\x_k \neq \NE_k$. 
    

    The equilibrium concept which we are primarily interested in with regards to Q-Learning is the Quantal Response Equilibrium (QRE) \cite{camerer:bgt}, which acts as an equilibrium concept when agents have bounded rationality. 

    \begin{definition}[Quantal Response Equilibrium (QRE)] A joint mixed strategy $\NE \in \Delta$ is a \emph
     {Quantal Response Equilibrium} (QRE) if, for all agents $k$ and all actions $i \in \actionset{k}$
        \begin{equation*}
            \NE_{ki} = \frac{\exp(r_{ki}(\NE_{-k})/T_k)}{\sum_{j \in \actionset{k}} \exp(r_{kj}(\NE_{-k})/T_k)}
        \end{equation*}
    \end{definition}
    where $T_k$ denotes the \emph
 {exploration rate} of the agent: low values of $T_k$ correspond to a higher tendency to exploit the best performing action. This, purely rational behaviour, corresponds to the \emph{Nash Equilibrium} in the limit $T_k \rightarrow 0$.  Higher values of $T_k$, meanwhile, implies a higher exploration rate. The two equilibria concepts are
 related by the following results

    \begin{proposition}[\cite{melo:qre} Proposition 2] Consider a game $\Gamma = (\agentset, (\actionset{k}, u_k)_
     {k \in \agentset})$ and take any         $T_k > 0$. Define the \emph{perturbed game} $\Gamma^H = (\agentset,
     (\actionset{k}, u^H_k{k \in \agentset})$ with the modified payoffs
        \begin{equation*} u_k^H(\x_k, \x_{-k}) = \langle \x_k, r_k(\x_
         {-k}) \rangle - T_k \langle \x_k, \ln \x_k \rangle
        \end{equation*}
        Then $\NE \in \Delta$ is a QRE of $\Gamma$ if and only if it is an NE of $\Gamma^H$.
    \end{proposition}

    Finally, in the application of our results (specifically Lemma \ref*{thm::StabilityUnique}), we will consider
    the \emph{influence bound} of a game, which gives a notion of its \emph{size}. Formally, we apply the definition
    from \cite{melo:qre}
    
    \begin{definition}[Influence Bound] Consider a finite normal form game $\Gamma$, the \emph{influence bound} $\delta$
     is given by
        \begin{equation*}
            \delta = \max_{k \in \mathcal{N}, i \in S_k, s_{-k}, \Tilde{s}_{-k} \in S_{-k}} \{ |r_{ki}(s_{-k}) - r_{ki}
             ( \Tilde{s}_{-k})| \}
        \end{equation*}
        where the pure strategies $s_{-k}, \Tilde{s}_{-k} \in S_{-k}$ differ only in the strategies of one agent $l \neq
        k$.

    \end{definition}

    Since $|r_{ki}(s_{-k}) - r_{ki}( \Tilde{s}_{-k})|$ measures the change in reward to agent $k$ for playing action $i$
    due to a change the other players' actions, the influence bound $\delta$ can be thought of as a measure of the
    maximum influence (in terms of reward) that any agent could receive from their opponents. In Section \ref{sec::Experiments}, we consider two games and show how their influence bound can be readily determined.

    \subsection{Learning Model}

    The principal multi agent learning model that we analyse in this study is the \emph{smooth Q-Learning} (QL) dynamic \cite
    {tuyls:qlearning} which is foundational in economics \cite{camerer:ewa} and artificial intelligence \cite
    {sutton:barto}. In \cite{tuyls:qlearning,sato:qlearning} a continuous time approximation of the Q-Learning
    algorithm was found which accurately captures the behaviour of learning agents through the following ODE:
        \begin{equation} \label{eqn::QLDynamics}
            \frac{\dot{x}_{ki}}{x_{ki}}=r_{ki}\left(\x\right)-\left\langle \x_{k}, r_{k}(\x) \right\rangle - T_{k}\left
             (\ln x_{ki}-\left\langle \x_{k}, \ln \x_{k}\right\rangle \right)
        \end{equation}
  A full review of this
 dynamic is beyond the scope of this study, however a derivation of (\ref{eqn::QLDynamics}) can be found in \cite
 {piliouras:zerosum} as well as the result that fixed points of the Q-Learning dynamic correspond to the QRE of the
 game $\Gamma$. 
        
        It was shown in \cite{piliouras:potential} that the transformation between the game $\Gamma$ and the perturbed
        game $\Gamma^H$ as described in Sec.~\ref{sec::GameModel} also relates the Q-Learning Dynamics to the the well
        studied \emph{replicator dynamics} (RD) \cite{smith:replicator,hofbauer:egd}, in the following manner.
         
        \begin{lemma}[\cite{piliouras:potential}]\label{lem::QLRD} Consider the game $\Gamma = (\agentset, (\actionset
         {k}, u_k)_{k \in \agentset})$ and some $T_k > 0$. Then the Q-Learning dynamics (\ref*{eqn::QLDynamics}) can be
         written as
            \begin{align}
                \frac{\dot{x}_{ki}}{x_{ki}}&=r^H_{ki}\left(\x\right)-\left\langle \x_{k}, r^H_{k}(\x) \right\rangle 
            \end{align} where $r^H_{ki}\left(\x\right) = r_{ki}(\x) - T_k\,(\ln x_{ki} + 1)$. In particular, (\ref*
             {eqn::QLDynamics}) corresponds to RD in the perturbed game $\Gamma^H = (\agentset, (\actionset{k}, u_k^H)_
             {k \in \agentset})$.
        \end{lemma}

        \paragraph{Follow the Regularised Leader} It is known \cite{mertikopoulos:reinforcement,flokas:donotmix} that RD
         is derived as a particular instance of the \emph{Follow the Regularised Leader} (FTRL) dynamic \cite
         {shalev:online}. In essence, FTRL requires that the agents maximise their cumulative payoff up to the current
         time $t$. However, it also imposes a regularisation on the agents' actions which softens the $\arg\max$
         function. More formally we have that for every agent $k$,
        \begin{eqnarray} \label{eqn::FTRL}
            \y_k(t) & = & \y_k(0) + \int_{0}^t \mathbf{r}_k(\x(s)) \, ds \nonumber \\
            \x_k(t) & = & Q_k(\y_t) := \arg\max_{\x_k \in \Delta_k} \left\{ \langle \x_k, \y_k \rangle - h_k
             (\x_k) \right\} 
        \end{eqnarray}
        To make RD compatible with FTRL, we make the following assumption on the regularisers
        \begin{assumption} \label{ass::regulariser} For every agent $k$, the regulariser $h_k$ is:
            \begin{enumerate}
                \item Continuously differentiable with differential $\nabla h_k$ which itself is Lipschitz on
                 $\Delta_k$, with constant $L_k$.
                \item \textit{Steep} in that $||\nabla h_k(\x)|| \rightarrow \infty$ as
                 $\x \rightarrow \partial \Delta$ (c.f. \cite{flokas:donotmix}).
                \item Strongly convex on $\Delta_k$, i.e. there exists $\kappa$ such that, for any
                 $\x_k, \y_k \in \Delta_k$
            \begin{equation*}
                \langle \nabla h_k(\y_k), \x_k - \y_k \rangle \leq h_k(\y_k) - h_k(\x_k) - \frac{\kappa}
                 {2} ||\x_k - \y_k||^2
            \end{equation*}
            \end{enumerate}
        \end{assumption}
        
%
         The choice of $h_k$ and $r_k$ will, of course, depend on the application scenario. Our interest in this work is
         to consider the case $h_k(x_k) = \sum_i x_{ki}\, \ln x_{ki}$, which satisfies Assumption \ref
         {ass::regulariser}, and from which RD is derived. As mentioned, smooth Q-Learning describes RD in the
         perturbed game $\Gamma^H$. This perturbation will be instrumental in proving our results on convergence.

    \subsection{Variational Inequalities and Game Theory}

    In this work we will examine game-theoretic concepts through the lens of Variational Inequalities \cite
    {facchinei:VI}. This branch of research modifies the problem of finding a Equilibrium of a game to that of finding
    a solution to a variational inequality, which is defined as follows.

    \begin{definition}[Variational Inequalities] Consider a set $\mathcal{X} \subset \R^d$ and a map $F \, : \mathcal
     {X} \rightarrow \R^d$. The Variational Inequality $VI(\mathcal{X}, F)$ is given as
            \begin{equation}\label{eqn::VIdef}
                \langle \x - \NE, F(\NE) \rangle \geq 0, \hspace{0.5cm} \text{ for all } \x \in \mathcal{X}.
            \end{equation}
        We say that $\NE \in \mathcal{X}$ belongs to the set of solutions to a variational inequality $VI(\mathcal
        {X}, F)$ if it satisfies (\ref{eqn::VIdef}).
        
    \end{definition}
    
    In this work, we will be considering the state space $\mathcal{X} = \Delta = \times_k \Delta_k$ alonside the map
    $F \, : \Delta \rightarrow \R^{Nn}$ defined as 
        \begin{equation*} F(\x) = \left( F_k(\x) \right)_{k \in \agentset} = \left( - r_k(\x) \right)_{k \in \agentset}
        \end{equation*}
        This map is sometimes referred to as the \emph{pseudo-gradient} of the game \cite{tatarenko:monotone}. Its
        properties allow for conditions to be found under which the equilibrium $\NE$ is unique. To illuminate these
        conditions, we have the following definition.
        \begin{definition}[Weighted Monotone Game] \label{def::weightedmon} A game $\Gamma$ with a continuous
         pseudo-gradient $F$ is
                \emph{weighted monotone} if there exist positive constants $w_1, \ldots, w_d$ such that, for all
                 $\x, \y \in \Delta$,
                \begin{equation} \label{eqn::weightedmonotone}
                    \langle \x - \y, F(\x; \w) - F(\y; \w) \rangle \geq 0
                \end{equation}
                where $F(\x, \w) = (w_1 F_1, \ldots, w_d F_d)^\top$.

        \end{definition}

        Naturally, the game is \emph{weighted strictly monotone} if, for all $\x \neq \y \in \Delta$, inequality~(\ref
        {eqn::weightedmonotone}) is strict. 
        
        Through these elements, it is possible to explore the nature of the equilibrium of a game from a variational
        perspective (see for instance \cite{facchinei:VI,parise:network,melo:qre}). The seminal results from this
        analysis, outlined in depth in \cite{facchinei:VI} are:
        
        \begin{lemma} \label{lem::VINE} Consider a game $\Gamma$ with pseudo-gradient $F(\x)$. Then, $\NE \in \Delta$ is
         an Equilibrium of $\Gamma$ if and only if it satisfies
        \begin{equation*}
            \langle \x - \NE, F(\x) \rangle \geq 0 \hspace{0.5cm} \text{ for all } \x \in \Delta.
        \end{equation*}
    \end{lemma}
    
    \begin{lemma} \label{lem::UniqueNE} If $\Gamma = (\agentset, (\actionset{k}, u_k)_{k \in \agentset})$ is a strictly
     monotone game, it has a unique Equilibrium $\NE \in \Delta$.
    \end{lemma}
    
    In our work, we will leverage the fact that both of these Lemmas extend readily to the weighted case (as all weights
    are assumed non-negative). By applying these properties, it is our goal to examine whether the equilibrium will be
    reached by learning agents. 

     \section{Learning in Weighted Monotone Games} \label{sec::Results}

    We have two main results in this section. The first is that multi-agent Q-Learning can achieve convergence to an
    equilibrium in \emph{any} game through sufficient exploration, parameterised by $T_k$. Following this result, we
    consider the optimality of convergence to an equilibrium. Namely we show a sufficient condition (convergence in
    time-average) in which non-fixed point behaviour (e.g. cycles) outperforms the equilibrium in terms of payoff.
    
    \subsection{Convergence through Sufficient Exploration}
    
    Without further delay, we state our first main result

    \begin{theorem} \label{thm::StabilityUnique} Let $\delta$ be the influence bound of an arbitrary game $\Gamma$ with
     $N$ players. Then, Q-Learning converges to the unique QRE $\NE$ if, for all agents $k$,
        \begin{equation} \label{eqn::UniqueQRE} T_k > \delta (N - 1)
        \end{equation}
    \end{theorem}

    This result provides a general condition from which the convergence to the QRE can be guaranteed in any game. As one
    might expect, the amount of exploration (the size of $T_k$) required to achieve this will be influenced by the size
    of the game (parameterised by $\delta$) and the number of players in the game (parameterised by $N$). An
    interesting point to note is that this results supports that of \cite{sanders:chaos} in which it was shown that
    non-fixed point behaviour is more likely with low exploration rates, if the number of players is increased. 
    
    \begin{proof}[Proof Sketch]

        Theorem \ref{thm::StabilityUnique} relies on two points: first, the perturbed game $\Gamma^H$ is weighted strictly
        monotone (which gives uniqueness of $\NE$) if $T_k$ satisfies (\ref{eqn::UniqueQRE}), and second, the QL dynamics
        converges to a fixed point in such games. The first point is shown through \cite{melo:qre} Theorem 1 so we
        focus on the latter.

        We achieve this in the following manner. First we show that, along trajectories of FTRL in weighted strictly
        monotone games, the distance to the equilibrium $\NE$ is decreasing. From this, the convergence of replicator
        in strictly monotone games is immediate. The reader will recall that Q-Learning describes replicator dynamics
        in a perturbed game. We show that, if the original game $\Gamma$ is weighted monotone, then the corresponding
        perturbed game $\Gamma^H$ is weighted strictly monotone. Putting all this together yields the convergence of QL
        in weighted monotone games.
     
        \paragraph*{Step 1: Convergence of FTRL} In order to show that the distance to the equilibrium $\NE$ is being
         decreased, we must first define what we mean by \emph{distance}. We do this through the \emph
         {Bregman Divergence} \cite{shalev:online}. 

    \begin{definition} Consider a set of functions $h_k : \Delta_k \rightarrow \R$ and a set of positive scalars $\w =
     (w_1, \ldots, w_N)$ with $k \in \agentset$. The \emph{Weighted Bregman Divergence} induced by $h = (h_k)_
     {k \in \agentset}$ between a set of probability vectors $\x, \y \in \Delta$ with $\x = (\x_k)_
     {k \in \agentset}$, $\y = (\y_k)_{k \in \agentset}$ is given by
        \begin{align*} W_{B}(\x || \y ; h) &= \sum_k w_k  D_{B}(\x_k || \y_k ; h) \\ 
              &= \sum_k w_k \left( h_k(\y_k) - h_k(\x_k) - \langle \nabla h_k(\x_k), \y_k - \x_k \rangle \right).
        \end{align*}

    \end{definition}

    \begin{remark} In particular, with the choice $h_k(x_k) = \sum_i x_{ki} \ln x_{ki}$, the Weighted Bregman Divergence
     corresponds to the \emph{Weighted Kullback-Leibler (KL) Divergence} defined by
        \begin{equation} W_{KL}(\x || \y) = \sum_{k \in \agentset} w_k D_{KL}(\x_k || \y_k) = \sum_
         {k \in \agentset} w_k \sum_{i \in \actionset{k}} x_{ki} \ln \frac{x_{ki}}{y_{ki}}.
        \end{equation}
    \end{remark}

    \begin{theorem} \label{thm::FTRLGlobalConv} Consider a weighted strictly monotone game $\Gamma$. If each agent
     follows an FTRL algorithm whose regulariser satisfies Assumption \ref{ass::regulariser} then, for any initial
     condition $\x(0)$, $\x(t)$ minimises the weighted Bregman divergence towards the unique equilibrium $\NE$.
    \end{theorem}

     \begin{remark} The question of convergence to the equilibrium $\NE$ can also be established through Theorem 4.9
      of \cite{mertikopoulos:finite} and showing that the assumptions required by their theorem are met by
      Assumption \ref{ass::regulariser}, which we do in the Appendix. However, in order to progress towards convergence
      of Q-Learning in \emph{weighted monotone games} (i.e. to lift the strictness requirement), we take the extra step
      of showing that the \emph{Bregman Divergence} is decreasing along trajectories.
        \end{remark}

    \paragraph*{Step 2: Convergence of RD} Using the specialisation $h_k(x_k) = \sum_i x_{ki} \ln x_{ki}$, and the
     relation between the Bregman and KL Divergences from , Theorem \ref*
     {thm::FTRLGlobalConv} implies that $W_{KL}(\NE || \x(t))$ is a Lyapunov function for RD in weighted strictly
     monotone games.
    
    \begin{corollary} \label{corr::RDConvergence} Consider a weighted strictly monotone game $\Gamma$. Then trajectories
     $\x(t)$ under the replicator dynamics minimise the KL-Divergence to the unique Equilibrium.
    \end{corollary}

    \begin{remark} Corollary \ref{corr::RDConvergence} mirrors the result found by \cite{sorin:composite} (Prop. 4.5),
     which also showed that the KL-divergence decreases. Here, we show that this is a special case of Theorem \ref
     {thm::FTRLGlobalConv}.
    \end{remark}

    \paragraph{Step 3: Convergence of Q-Learning} Finally, we recognise that the same transformation which takes
     $\Gamma$ to $\Gamma^H$ also takes Q-Learning to RD (c.f. \cite{piliouras:potential} Lemma 3.1). Putting this all
     together, if $T_k$ satisfies (\ref{eqn::UniqueQRE}), the perturbed game $\Gamma^H$ is strictly monotone, which
     yields the convergence of RD to the unique equilibrium $\NE$. This immediately gives the convergence of QL in the
     original game $\Gamma$.
    \end{proof}
    
    \subsection{Convergence through Arbitrary Exploration}
    
    Whilst Theorem \ref{thm::StabilityUnique} gives a condition to achieve convergence through sufficient exploration,
    it is known that there are games for which convergence to a QRE can be achieved with any non-zero exploration rate
    $T_k$. These games are: weighted potential games \cite{piliouras:potential} and weighted zero sum polymatrix
    games \cite{piliouras:zerosum}. In this section, we point out that, under suitable game structures, the reduced requirement on 
    exploration in these games stems from the fact that they are both examples of \emph{monotone} games. We can, therefore, apply the
    following Theorem to yield convergence of QL in such games.

    \begin{theorem}\label{thm::QLConvergence} If the game $\Gamma$ is weighted monotone then, for any $T>0$, Q-Learning
     converges to the unique QRE.
    \end{theorem}
    
     The proof of this statement, presented in full in the Appendix, lies in recognising that the transformation which takes $\Gamma$ to $\Gamma^H$ is an additive term of the form
    
    \begin{equation*}
        -T_k (\ln x_{ki} + 1)
    \end{equation*}
    
    which is a strictly monotone function in $\x_k$. As such, if $\Gamma$ is already weighted monotone then, for any positive value of $T_k$, $\Gamma^H$ must be weighted strictly monotone.
    
     \paragraph*{Weighted Potential Games} Our first application concerns the dynamics of learning in weighted potential
      games. These concern the cooperative setting in which the game admits a global function over all agents'
      strategies. Formally, a game is called a weighted potential game if there exists a function
      $U \, : \Delta \rightarrow \R$ and positive weights $w_1, \ldots, w_N > 0$ such that, for each player
      $k \in \agentset$, $U(\x_k, \x_{-k}) - U(\y_k, \x_{-k}) = w_k \left(u_k(\x_k, \x_{-k}) - u_k(\y_k, \x_
      {-k}) \right)$ for all $\x_k, \y_k \in \Delta_k$ and all $\x_{-k} \in \Delta_{-k}$.

    \begin{lemma}\label{lem::PotentialConv} Consider a weighted potential game $\Gamma$ with concave potential $U
     (\x)$. Then, for any $T_K > 0$, the Q-Learning dynamics converges to the unique QRE $\NE$.
    \end{lemma}

    Of course, it must be noted that, in many games, the potential is not concave and so require a different approach
    towards showing convergence. \cite{piliouras:potential} performs such an analysis and also considers the geometry
    of the multiple QRE which can exist due to the non-concavity of the potential.

    \paragraph*{Weighted Zero Sum Polymatrix Games}

    In this application, we consider the competitive setting through the \emph{weighted zero sum network polymatrix game}.
    Formally, a network polymatrix game \newline $\Gamma = \left( (\agentset, \edgeset), (\actionset{kl}, A^{kl})_{
    (k, l) \in \edgeset} \right)$ includes a network $(\agentset, \edgeset)$ in which $\edgeset$ consists of pairs of
    agents $k, l \in \agentset$ who are connected in the network. Each edge is imbued with payoff matrices $A^
    {kl}$ which denotes the reward to agent $k$ against agent $l$. The total payoff received by agent $k$, then, is
    \begin{equation} u_k(\x_k, \x_{-k}) = \sum_{(k, l) \in \edgeset} \langle \x_k, A^{kl}\x_l \rangle
    \end{equation}
    A game is called \emph{weighted zero sum network polymatrix game} if there exists positive constants $w_1, \ldots, w_N$ such
    that, for all $\x \in \Delta$
    \begin{equation}
        \sum_{k \in \agentset} w_k \langle \x_k, A^{kl}\x_l \rangle = 0
    \end{equation}

    \begin{lemma}\label{lem::NZSGConv} Consider a weighted zero sum network polymatrix game $\Gamma$. The unique QRE $\x$ is
     globally asymptotically stable under QL for any $T_k > 0$.
    \end{lemma}

    The convergence of Q-Learning in this class of games is also proven through an alternate proof in \cite
    {piliouras:zerosum}, in which the weighted KL-divergence was also found as a Lyapunov function. Lemma \ref
    {lem::NZSGConv} proves this point by showing that weighted network polymatrix zero sum games fall under the more general class
    of weighted monotone games. In such a manner, the results of \cite{piliouras:zerosum}, can be extend to all weighted
    monotone games.
    
    \subsection{Learning Outperforms the Equilibrium}

    In our second result, we consider the optimality of exploration as a means of reaching the equilibrium. We know, for
    instance, that RD, which corresponds to QL with zero exploration rates $T_k$, will not converge in a merely
    monotone game. However, in this case, we know from \cite{mertikopoulos:finite} that the \emph{time-average} of the
    trajectory does reach the equilibrium. In this case we are presented with the following question
    
    \textit{Should we apply non-zero exploration rates $T_k$ to ensure convergence to an equilibrium, or allow the
     trajectory to remain non convergent?}
    
    We answer this question by considering the \emph{payoff performance} of the dynamic, where performance is measured
    by the Time Averaged Social Welfare (TSW)
    \begin{eqnarray} 
    TSW & = & \lim_{t \rightarrow \infty}\frac{1}{t} \int_0^t SW(\x(s)) ds. \label{eqn::TSW}
    \end{eqnarray} 
    
    where $SW(\x) = \sum_k u_k(\x_k, \x_{-k})$. Intuitively, the Social Welfare $SW(\x)$ measures the sum of payoffs received by all agents at some state
    $\x \in \Delta$. In (\ref{eqn::TSW}), this quantity is averaged along trajectories of the learning dynamic. From
    the following theorem we see that, in monotone games, the social welfare is higher if agents play according to
    FTRL (and therefore do not converge), than if they were to play according to the equilibrium.

    \begin{theorem} \label{thm::PayoffPerformance} 
    Consider a polymatrix game $\Gamma = ( (\agentset, \edgeset),
     (\actionset{kl}, A^{kl})_{(k, l) \in \edgeset} )$. It is the case that
        \begin{equation*} u_{k}(\x_k, \x_{-k}) = \sum_{(k, l) \in \edgeset} \langle \x_k, A^{kl} \x_l \rangle.
        \end{equation*} Supposing that the learning dynamic is such that $\lim_{t \rightarrow \infty} \frac{1}
         {t} \int_0^t R(s) \, ds = 0$, where $R(t) = \sum_k R_k(t)$ and 
        \begin{align*} R_k(t) &= \max_{\x_k' \in \Delta_k} \int_0^t \left [u_k(\x_k', x_{-k}
         (s)) - u_k(\x_k(s), \x_{-k}(s)) \right] \, ds,
        \end{align*} 
        is the \emph{regret} of agent $k$ \cite{shalev:online} and supposing also that the time averaged strategy $\mu(t) := \frac{1}{t} \int_0^t \x(s) \, ds$
         along the learning dynamics converges to the equilibrium $\NE$, then TSW is asymptotically greater than or
         equal to the Social Welfare of the equilibrium $\NE$, i.e. $TSW \geq SW(\NE)$
    \end{theorem}
    
    \begin{remark} In particular this result holds for FTRL dynamics in monotone games, for which the time average
     approaches the equilibrium $\NE$ asymptotically \cite{mertikopoulos:finite}. In addition, FTRL is a \emph
     {no-regret} dynamic, in that $\lim_{t \rightarrow \infty} \frac{1}{t} \int_0^t R_k(s) \, ds = 0$ for all agents
     $k$.
    \end{remark}
    
    \begin{proof} 
    

        Define $\mu_k(t) \in \Delta_k$ as the time averaged strategy of agent $k$, i.e.
        \begin{equation} \label{eqn::timeaveragestrat}
            \mu_k(t) = \frac{1}{t} \int_0^t \x_k(s) \, ds.
        \end{equation}
        Then,
        \begin{align*} R_k(t) &= \max_{\x_k' \in \Delta_k} \int_0^t u_k(\x_k', x_{-k}(s))  \, ds - \int_0^t u_k(\x_k
         (s), \x_{-k}(s)) \, ds \\
            &\geq \int_0^t u_k(\mu_k(t), x_{-k}(s))  \, ds - \int_0^t u_k(\x_k(s), \x_{-k}(s)) \, ds \\
            &= \int_0^t \mu_k(t) \cdot \sum_{(k, l) \in \edgeset} A^{kl} \x_l(s) \, ds - \int_0^t \x_k(s) \cdot \sum_{
             (k, l) \in \edgeset} A^{kl} \x_l(s) \, ds
        \end{align*}
        Then, dividing by $t$ and applying (\ref{eqn::timeaveragestrat}) to $\x_l$, the inequality reads
        \begin{align*}
            \frac{1}{t} R_k(t) \geq \mu_k(t) \cdot \sum_{(k, l) \in \edgeset} A^{kl} \mu_l(t) -  \frac{1}
             {t} \int_0^t \x_k(s) \cdot \sum_{(k, l) \in \edgeset} A^{kl} \x_l(s) \, ds 
        \end{align*}
        Taking the sum over all agents $k$ and defining $R(t) = \sum_k R_k(t)$ we have
        \begin{small}
        \begin{align*}
            \frac{1}{t} R(t) &\geq \sum_k \Big( \mu_k(t) \cdot \sum_{(k, l) \in \edgeset} A^{kl} \mu_l(t)  - \frac{1}
             {t} \int_0^t \x_k(t) \cdot \sum_{(k, l) \in \edgeset} A^{kl} \x_l(s) \, ds \Big)\\
            &= SW(\mu(t)) - \frac{1}{t} \int_0^t SW(\x(s))  \, ds 
        \end{align*}
        \end{small}
        \begin{align*}
            &\implies \frac{1}{t} \int_0^t SW(\x(s))  \, ds \geq SW(\mu(t)) - \frac{1}{t} R(t) \\
            &\implies \lim_{t \rightarrow \infty}\frac{1}{t} \int_0^t SW(\x(s))  \, ds \geq \lim_
              {t \rightarrow \infty} SW(\mu(t))
        \end{align*} where, in the final equality, we use the assumption that $\lim_{t \rightarrow \infty} \frac{1}{t} R
         (t) = 0$. Applying now the assumption that $\mu(t) \rightarrow \NE$, we have the final result that
        \begin{equation*}
            \lim_{t \rightarrow \infty}\frac{1}{t} \int_0^t SW(\x(s))  \, ds  \geq SW(\NE).
        \end{equation*}
    \end{proof}
    
    \begin{remark}
     We point out here that Theorem \ref{thm::PayoffPerformance} defines payoff performance as a \emph{time-average} and as a sum over all agents. As such, it is possible that a particular agent is losing out on payoff consistently throughout learning. Similarly, it is possible that at some time $t$, all agents are performing worse than the equilibrium. Finally, the result concerns asymptotic behaviour, so it could be the case that the equilibrium initially outperforms the learning dynamic. However, Theorem \ref{thm::PayoffPerformance} shows that, eventually, the cumulative, time-averaged performance (\ref{eqn::TSW}) will be better than the equilibrium.
    \end{remark}
    
    Similar results on the performance of non-fixed point learning dynamics have been found for Fictitious Play \cite{ostrovski:payoff},
    Replicator Dynamics \cite{kleinberg:nashbarrier} and no-regret learning \cite{anagnostides:last-iterate}. In the latter most of these, TSW was also applied as a measure of payoff performance and it was found that the payoff in discrete time \emph{online mirror descent} outperforms the Nash Equilibrium. Theorem \ref{thm::PayoffPerformance} adds to the body of literature studying payoff performance of learning by showing a sufficient condition (convergence in time-average) under which any no-regret learning algorithm (including Fictitious Play and Replicator) will outperform the equilibrium.

    \subsection{Discussion of Results} 

    Theorem \ref{thm::StabilityUnique} presents an concrete approach for achieving convergence to a unique equilibrium
    in any game. This allows us to abstract beyond the typical arena of study: namely weighted potential or zero sum
    network games. In fact, we showed convergence in both of these cases due to the fact that they satisfy the
    monotonicity assumption and, therefore, satisfy the assumptions of Theorem \ref{thm::QLConvergence}. 

    Though general in nature, these results are not without their limitations. They rely on the assumption of a discrete
    action set, so that agent strategies all evolve on $\Delta$. This allows us to assume the existence of an
    Equilibrium, through the compactness of $\Delta$. However, generalising to arbitrary continuous action sets would
    widen the range of applications which our work encompasses. In addition, Theorem \ref*{thm::QLConvergence} is
    derived for continuous time QL. This is a reasonable stance to take as it has been shown repeatedly that continuous
    time approximations of algorithms provide a strong basis for analysing the algorithms themselves
    \cite{hofbauer:zerosum,tuyls:qlearning}. However, the accuracy of discrete time algorithms is always dependent on
     parameters, most notably step sizes. Such an analysis of the discrete variants presents a fruitful avenue for
     further research.

    To qualify Theorem \ref*{thm::StabilityUnique}, we point out that it does not give any indication regarding the
    performance of the system, merely its behaviour. Furthermore, its success relies on the increasing exploration
    rates of the agents and, therefore, at the cost of their exploitation. We showed in Theorem \ref
    {thm::PayoffPerformance} that, under certain conditions, agents following a non-convergent dynamic may actually
    outperform the equilibrium in terms of payoff. Our experiments monitor this effect further by showing that, in
    certain games, as $T_k$ is increased the performance of the system decreases.
       \section{Experiments} \label{sec::Experiments}

    \begin{figure*}[t!]
        \begin{subfigure}{\textwidth}
            \includegraphics[width=\columnwidth]{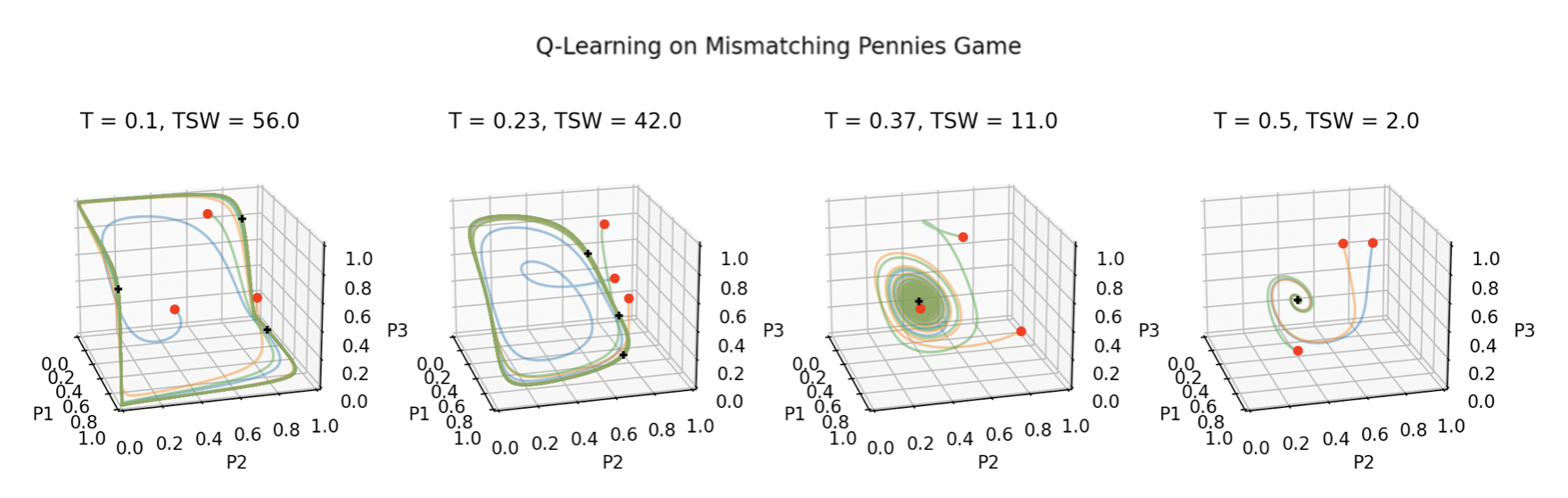}
        \end{subfigure}

        \begin{subfigure}{\textwidth}
            \includegraphics[width=\columnwidth]{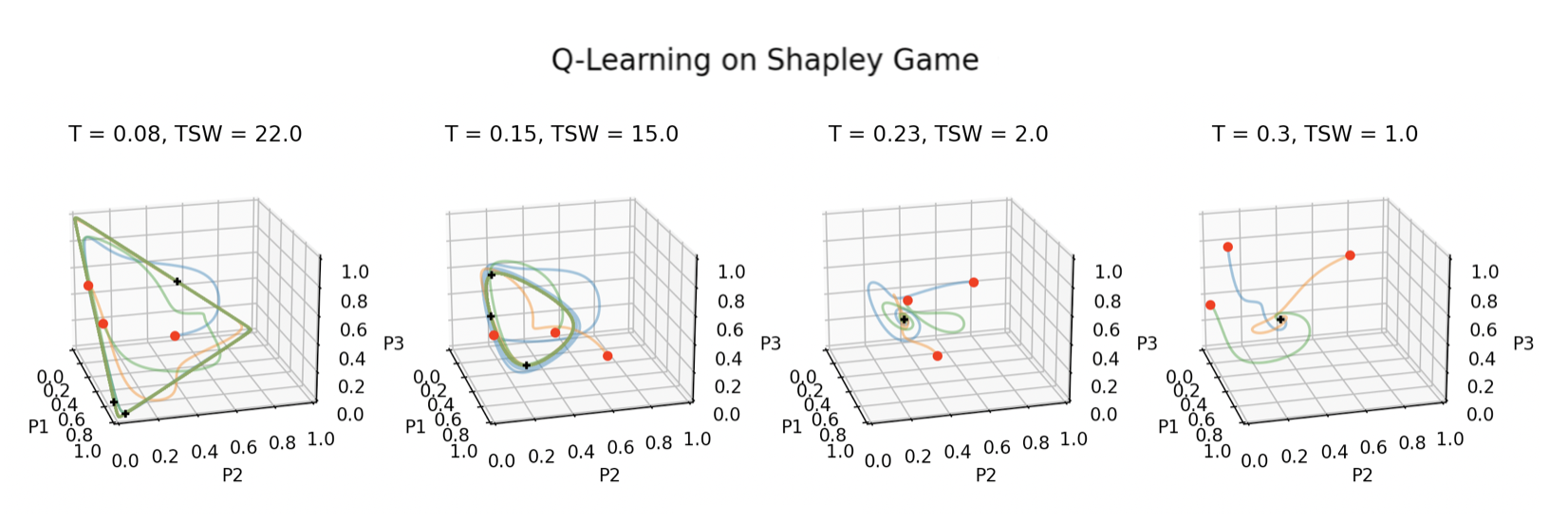}
        \end{subfigure}        
        \caption{\label{fig::Performance} Trajectories of Q-Learning, red dots indicate initial conditions whilst black
         markers indicate final positions. TSW is averaged over all three initial conditions.}
    \end{figure*}

    \begin{figure*}[t!]
        \begin{subfigure}{0.45 \textwidth}
            \includegraphics[width= \textwidth]{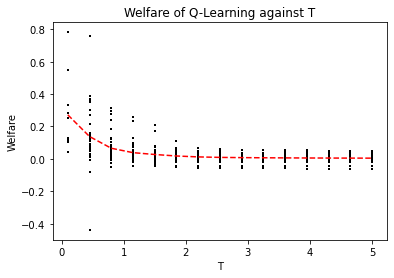}
            \caption{\label{fig::Welfarep5n5}}
        \end{subfigure}
        \begin{subfigure}{0.45 \textwidth}
            \includegraphics[width=\textwidth]{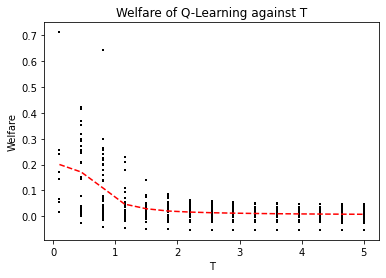}
            \caption{\label{fig::Welfarep7n5}}
        \end{subfigure}        
        \caption{\label{fig::Welfare} Normalised TSW against $T$ for 35 randomly generated games. TSW is normalised to
         lie between [-1, 1] in each game. Results are averaged over 10 initial conditions. The red line denotes the
         mean TSW across all 35 games. a. Five Players, Five Actions. b. Seven Players, Five Actions}
    \end{figure*}

    In the following experiments, we consider the limitation discussed in the previous section. Namely, we test to see
    whether the optimality of the learning algorithm may decrease as $T_k$ is increased. To do this we consider two
    examples: the Mismatching Pennies and the Shapley games. The former, proposed in \cite{kleinberg:nashbarrier}, is
    composed of a network of three agents, each equipped with two actions, Heads and Tails. The payoff given to each
    agent $k$ is given by
    \begin{align*} u_k &= u_k(\x_k, \x_j) = \x_k \A \x_{k-1} \\
        \A &= \begin{pmatrix} 0 & 1\\ M & 0 \\
        \end{pmatrix}, \hspace{0.2cm} M \geq 1.
    \end{align*}
    For the latter \cite{shapley:twoperson}, we examine a network variation on the original two player game towards a
    three player network variant in which payoffs, parameterised by $\beta \in (0, 1)$ are defined as
    \begin{align*} u_k &= u_k(\x_k, \x_{-k}) = \x_k \A \x_{k-1} + \x_k \B^\top \x_{k+1}\\ A&=\left(\begin{array}
     {ccc} 1 & 0 & \beta \\
            \beta & 1 & 0 \\ 0 & \beta & 1
            \end{array}\right), \, B=\left(\begin{array}{ccc}
            -\beta & 1 & 0 \\ 0 & -\beta & 1 \\ 1 & 0 & -\beta
            \end{array}\right),
    \end{align*}
    \begin{remark} In the case of a network polymatrix game, the influence bound of the game
     is
    \begin{equation*}
        \max_{k \in \mathcal{N}, i \in S_k, s_{-k}, \Tilde{s}_{-k} \in S_{-k}} |\left( A^k \right)_{i, s_{-k}} - \left
         ( A^k \right)_{i, \Tilde{s}_{-k}}|.
    \end{equation*}
    In other words, it is the maximum difference between any row elements across the payoff matrices for all agents. In the Mismatching game, the influence bound is $M$ whilst the Shapley game has influence bound $1 + \beta$.
    \end{remark}
    
    These games were analysed in \cite{kleinberg:nashbarrier} and \cite{ostrovski:payoff} respectively, and it was shown
    that while the given learning algorithm (RD in the former, Fictitious Play in the latter) did not converge to the
    NE, the agents actually received greater payoff through learning than they would have if they had only played the
    equilibrium strategy. The implication is that the agents were better off (in terms of payoff) by not converging to
    an equilibrium.

    In Fig.~\ref*{fig::Performance} we plot the trajectories of QL for varying choices of $T_k$ in each of the games.
    For the sake of simplicity, we enforce that all agents have the same $T_k$ so we drop the $k$ notation. The
    trajectories are displayed on the space $(0, 1)^3$ with each axis corresponding to the
    probability with which each agent plays their first action. Above each figure is displayed the choice of $T$ for
    which Q-Learning is run, as well as the Time-averaged Social Welfare (TSW) along the trajectory, given by (\ref
    {eqn::TSW}).
    
    Two points become immediately clear from Fig.~\ref*{fig::Performance}. The first is that as $T$ is increased, the
    dynamics break no longer cycle around the equilibrium but rather converge to a unique equilibrium. While this
    occurs, however, TSW is decreasing. In fact, even in the case of equilibriation, trajectories which take longer to
    reach the QRE gain a larger TSW. It is clear then, that it is in the agents' benefit if the dynamics remain
    unstable, at least as far as payoff is concerned.

    In Fig.~\ref*{fig::Welfare}, we move beyond these indicative examples by evaluating TSW on 35 randomly generated
    games as $T$ is increased. In order to accurately compare games with differnet payoff functions, we divide $\sum_
    {k} u_k(\x_k(t), \x_{k}(t))$ by the maximum possible cumulative payoff that an agent could receive in the game.
    This ensures that TSW remains within $[-1, 1]$ in all games. It is clear once again that, in general, TSW decreases
    as $T$ increases, i.e. as the game move towards more convergent behaviour. Of course, this does not hold in every
    game. However, the red line, which denotes the mean TSW across all games, suggests that this trend is the expected
    behaviour for a randomly selected game.

     \section{Conclusion}

        Our community has made strong strides in showing that online learning in games does not always reach an
        equilibrium. At the same time, the rising use of multiple interacting agents in machine learning applications
        necessitates placing guarantees on learning. In this paper, we make a step towards resolving this dichotomy by
        considering how the structure of a game, beyond the correlation between agent payoffs, affects online
        learning.

        Specifically, we considered the asymptotic convergence to unique fixed points through Q-Learning (QL). Our
        analysis shows that the convergence in this popular learning dynamic can be guaranteed through sufficient
        exploration on the part of all agents. We also subsume convergence results in co-ordination (potential) games
        and competitive (network zero sum) games for which any positive rate of exploration is required. We then
        consider the impact of convergence through the lens of payoff performance and show that no-regret algorithms
        will outperform the equilibrium in terms of payoff, so long as the time-average trajectory reaches an
        equilibrium. In our experiments we show that this behaviour holds for a large number of games. An interesting
        point for future work would be to develop an analytical understanding for how often non-convergent learning
        dynamics outperform the equilibria of the game. As our study has shown, convergence of dynamics is inextricably linked to exploration. As such, by studying the optimality of non-convergent dynamics, one may assess quantitatively the trade-off between exploration and exploitation.

        \subsection*{Acknowledgements}
        Aamal Hussain and Francesco Belardinelli are partly funded by the UKRI Centre for Doctoral Training in Safe and Trusted Artificial Intelligence (grant number EP/S023356/1). This research/project is supported in part by the National Research Foundation, Singapore and DSO National Laboratories under its AI Singapore Program (AISG Award No: AISG2-RP-2020-016), NRF 2018 Fellowship NRF-NRFF2018-07, NRF2019-NRF-ANR095 ALIAS grant, grant PIESGP-AI-2020-01, AME Programmatic Fund (Grant No.A20H6b0151) from the Agency for Science, Technology and Research (A*STAR) and Provost’s Chair Professorship grant RGEPPV2101.

    \section*{Appendix}

    \subsection*{Results on Learning in Games}

    A fundamental point to be noted of FTRL is that its dynamics evolve in the payoff space. To be able to translate this into a dynamical system on $\Delta$, we
    must consider the relation between $y_k$ and the corresponding state vector $x_k$. We do this through the following Lemma.
    
    \begin{lemma}[\cite{flokas:donotmix} Lemma B.4] \label{lem::steepness}
        $\x_k = Q_k(\y_k)$ if and only if there exist $\mu_k \in \R$ and $\nu_{ki} \in \R_+$ such that, for all $i \in S_k$ the following hold
        \begin{enumerate}
            \item $y_{ki} = \frac{\partial h_k}{\partial x_{ki}} + \mu_k - \nu_{ki}$
            \item $\nu_{ki} x_{ki} = 0$ 
        \end{enumerate}
        In particular, if $h_k$ is steep, then $\nu_{ki} = 0$ for all $i$.
    \end{lemma}
    
    The proof of this Lemma is in \cite{flokas:donotmix} and so we omit it here. We make use of this lemma in proving Theorem \ref*{thm::FTRLGlobalConv} which on the relation of the Bregman Divergence to trajectories generated by FTRL.

    We begin by considering the \emph{Fenchel coupling} generated by $h_k$ defined by
        \begin{equation*}
            F_k(\x, \y) = h_k(\x) + h_k^*(\y) - \langle \x, \y \rangle
        \end{equation*}
        In \cite{mertikopoulos:finite}, it was shown that, for regularisers who are strongly convex, the Fenchel coupling is a Lyapunov function for FTRL in strictly monotone games, provided the regulariser also satisfies the \emph{reciprocity} condition: for any $\x$ and any sequence $\y_n$
        \begin{equation*}
            Q_k(\y_n) \rightarrow \x \implies F_k(\x, \y_n) \rightarrow 0
        \end{equation*}
        The converse of this statement is already satisfied by the strong convexity of $h_k$. We begin by showing that regularisers which satisfy Assumption \ref{ass::regulariser} satisfy the reciprocity condition.

        \begin{lemma} \label{lem::reciprocity}
             Any regulariser $h_k$ which satisfies \ref{ass::regulariser} also satisfies that for any $\x \in \Delta$ and any sequence $(\y_n)_{n \in \mathbb{N}} \in \R^n$, 
             \begin{equation}
                Q_k(\y_n) \rightarrow \x \iff F_k(\x, \y_n) \rightarrow 0     
             \end{equation}
        \end{lemma}
    \begin{proof}
        Define, for any $n \in \mathbb{N}$, $\x_n = Q_k(\y_n) := \arg\max_{\x \in \Delta_k}  \left\{ \langle \x, \y_n \rangle - h_k(\x) \right\}$. 
        For the forward direction:
        \begin{align*}
            F_k(\x, \y_n) & = h_k(\x) + h_k^*(\y_n) - \langle \x, \y_n \rangle \\
            & = h_k(\x) + \left( \langle \x_n, \y_n \rangle - h_k(\x_n) \right) - \langle \x, \y_n \rangle \\
            & = h_k(\x) - h_k(\x_n) + \langle \x_n, \y_n \rangle - \langle \x, \y_n \rangle \\
            & = h_k(\x) - h_k(\x_n) + \langle \y_n, \x_n - \x \rangle
        \end{align*}
        By the convexity of $h_k$
        \begin{equation*}
            \langle \nabla h_k(\x), \x_n - \x \rangle \leq h_k(\x_n) - h_k(\x)
        \end{equation*}
        So that
        \begin{align*}
            F_k(\x, \y_n) & \leq - \langle \nabla h_k(\x), \x_n - \x \rangle + \langle \y_n, \x_n - \x \rangle \\
            & =  \langle \y_n - \nabla h_k(\x), \x_n - \x \rangle
        \end{align*}
        We apply Lemma \ref{lem::steepness} and Assumption \ref{ass::regulariser} to say that, since $\x_n = Q_n(\y_n)$,  $\y_n = \nabla h_k(\x_n) - \mu_k \mathbb{1}^\top$, so we are left with
        \begin{align*}
            F_k(\x, \y_n) & \leq \langle \nabla h_k(\x_n) - \nabla h_k(\x), \x_n - \x \rangle - \mu_k \langle \mathbb{1}, \x_n - \x \rangle \\
            & \leq ||\nabla h_k(\x_n) - \nabla h_k(\x)|| \cdot ||\x_n - \x|| \\
            & \leq L_k ||\x_n - \x||^2
        \end{align*}
        where the second inequality follows from the Cauchy-Schwartz Inequality and the final result follows from Assumption \ref{ass::regulariser}.1. Then, taking the limit as $n \rightarrow \infty$, $||\x_n - \x||^2 \rightarrow 0$ which, alongside the fact that $F_k(\x, \y) \geq 0$ for all $\x, \y$ gives the required result. This last fact follows directly from the Fenchel-Young Inequality.
        
        For the reverse direction, we apply Lemma 4.8 of \cite{mertikopoulos:finite} and Assumption \ref{ass::regulariser}.2.
    \end{proof}

      The advantage of making Assumption \ref{ass::regulariser}.1 is that it allows for the Bregman Divergence to be related to the Fenchel Coupling. This is formalised through the following extension of \cite{abe:MFTRL} Lemma C.1 to an arbitrary number of agents.
 
        \begin{lemma} \label{lem::Abe}
            For any agent $k$ following FTRL with regulariser $h_k$ who satisfies Assumption \ref{ass::regulariser}, 

            \begin{equation*}
                D_B(\x_k|| \x_k(t) ; h_k) = F_k(\x_k, \y_k(t)) =  
                h(\x_k) + h^*_k(\y_k(t))
                -  \langle \x_k, \y_k \rangle
            \end{equation*}
            for any $\x_k \in \Delta_k$
        \end{lemma}

    \begin{proof}
        Here, we extend Lemma C.1 of \cite{abe:MFTRL} to the general $N$-player game. Here, we recall again that $\x(t) = Q_k(\y(t)) = \arg\max_{\x \in \Delta_k}  \left\{ \langle \x, \y(t) \rangle - h_k(\x) \right\}$
        \begin{align*}
            D_B(\x||\x(t); h_k) - F_k(\x, \y(t)) & = h_k(\x) - h_k(\x(t)) - \langle \nabla h_k(\x(t)), \x - \x(t) \rangle - \left( h_k(\x) + h^*_k(\y(t)) - \langle \y(t), \x \rangle \right) \\
            & = - h_k(\x(t)) - \langle \nabla h_k(\x(t)), \x - \x(t) \rangle - \langle \y(t), \x(t) \rangle + h_k(\x(t)) + \langle \y(t), \x \rangle \\
            & = \langle \nabla h_k(\x(t)) - \y(t), \x - \x(t) \rangle\\
            & = 0
        \end{align*}
        where the final result follows from the fact that $\x(t) = Q_n(\y(t)) \iff \y(t) = \nabla h_k(\x(t)) - \mu_k \mathbb{1}^\top$.
    \end{proof}

    We use this to find, for each agent, the time derivative of the Weighted Bregman Divergence between a trajectory $\x(t)$ generated by FTRL and any other strategy $\y \in \Delta$ in terms of the pseudo-gradient of the game $F$. We formalise this through the following Lemma. 

        \begin{lemma} \label{lem::lyapFTRL}
            In any game $\Gamma$ with pseudo-gradient map $F$, let $\x(t)$ denote the trajectory generated by FTRL with regularisers $(h_k)_{k \in \agentset}$ and let $\y \in \Delta$.
            Then, for any positive set of weights $\w = (w_k)_{k \in \agentset}$. 
            
            \begin{equation} 
                \frac{d}{dt} D_{B}(\y_k , \x_k(t) ; h_k) = - \left( \langle \x_k - \y_k, F_k(\x) - F_k(\y) \rangle + \langle \x_k - \y_k, F_k(\y) \rangle \right)
            \end{equation}
        \end{lemma}
        
    \begin{proof}
        \begin{align}
            \frac{d}{dt} D_{B}(\NE_k || \x_k(t) ; h_k) &= \frac{d}{dt} \left( h_k(\NE_k) + h_k^*(\y_k(t)) - \langle \y_k(t), \NE_k(t) \rangle \right) \\
            &= \langle \nabla h^*_k(\y(t)), \dot{\y}(t) \rangle - \langle \dot{\y}_k(t), \NE_k \rangle\\
            &= \langle \nabla h^*_k(\y(t)) - \NE_k, \dot{\y}_k(t) \rangle\\
            &= \langle \x_k(t) - \NE_k, r_k(\x(t)) \rangle \label{eqn::Duality}\\
            &=  \langle r_k(\x(t)) - r_k(\NE), \x_k(t) - \NE_k \rangle + \langle r_k(\NE), \x_k(t) - \NE_k \rangle \nonumber \\
            &= - \left( \langle \x_k(t) - \NE_k, F_k(\x(t)) - F_k(\NE) \rangle + \langle \x_k(t) - \NE_k, F_k(\NE) \rangle \right) \nonumber
        \end{align}
        
        where $h^*_k$ is the convex conjugate of $h_k$ and \ref{eqn::Duality} holds due to \cite{shalev:online} (2.13).
        
    \end{proof}

    \begin{proof}[Proof of Theorem \ref{thm::FTRLGlobalConv}]

        Let $w_1, \ldots w_N > 0$ be constants for which $F$ is weighted strictly monotone.
        \begin{align}
            \frac{d}{dt} W_{B}(\NE || \x(t) ; h) &= \sum_k \frac{d}{dt} w_k D_{B}(\NE || \x(t) ; h) \\
            &= - \sum_k \left( \langle \x_k(t) - \NE_k, w_k F_k(\x(t)) - w_k F_k(\NE) \rangle + \langle \x_k(t) - \NE_k, w_k F_k(\NE) \rangle \right) \\
            &= - \left( \langle \x(t) - \NE, F(\x(t); \w) - F(\NE; \w) \rangle + \langle \x(t) - \NE, F(\NE; \w) \rangle \right) \\ & \leq 0
        \end{align}
        
        Applying the fact that $\NE$ is a Equilibrium, Lemma \ref*{lem::VINE} implies that $\langle \x - \NE, F(\NE; \w) \rangle \geq 0$. Since, in addition, $F$ is weighted strictly monotone, the final inequality holds. Equality holds if and only if $\x(t) = \NE$. Thus, $W_{B}(\NE || \x(t) ; h)$ is a strict Lyapunov function for FTRL, converging to $\NE$.
    \end{proof}

     \begin{proof}[Proof of Theorem \ref*{thm::StabilityUnique}]
        Theorem 1 of \cite{melo:qre} implies that the given inequality (\ref{eqn::UniqueQRE}) ensures the strong monotonicity of the perturbed game $\Gamma^H$. In such a
        case, the equilibrium $\NE$ is unique and, from Corollary \ref{corr::RDConvergence}, convergence of RD is guaranteed.
        Applying the relationship between RD and QL, the convergence of QL in $\Gamma$ to the unique QRE $\NE$ follows.
    \end{proof}

    \subsection*{Potential Games and Network Zero Sum Games}
    
    \begin{lemma}\label{lem::QLMonotone}
        If the game $\Gamma$ has a weighted monotone pseudo gradient $F$ then, for
        any $T>0$, the pseudo gradient of the perturbed game $\Gamma^H$ is weighted strictly
        monotone.
    \end{lemma}

    \begin{proof}
        Let us define $F$ (resp. $F^H$) as the pseudo gradient of $\Gamma$ (resp. $\Gamma^H$).
        We recall that the transformation between rewards in $\Gamma$ and $\Gamma^H$ is given by
        \begin{equation*}
            r_{ki}^H(\x) = r_{ki}(\x) - T_k (\ln x_{ki} + 1)
        \end{equation*}

        As such we have that, for any $\x, \y \in \Delta$ and any weighted $w_1, \ldots, w_N > 0$,
        \begin{align}
            \langle \x - \y, F^H(\x) - F^H(\y; \w) \rangle &= \sum_k \langle \x_k - \y_k, -w_k r_k^H(\x) - (-w_k r_k^H(\y)) \rangle \\
            &= \sum_k \langle \x_k - \y_k, w_k(-r_k(\x) + T_k\,\ln \x_k) - w_k (-r_k(\y) + T_k\,\ln \y_k)\rangle \\
            &= \langle \x - \y, F(\x; \w) - F(\y; \w) \rangle + \sum_k w_k T_k \langle \x - \y, \ln \x - \ln \y \rangle
        \end{align}

        Since: $\ln$ is a strictly monotone operator, $w_k > 0$ and, by assumption, we have that $F$ is weighted monotone, the result holds for any choice of $T_k > 0$.
    \end{proof}

    \begin{proof}[Proof of Theorem \ref*{thm::QLConvergence}]
        If $\Gamma$ has a weighted monotone pseudo-gradient then, by Lemma \ref*{lem::QLMonotone}, the perturbed game $\Gamma^H$ has a weighted strictly monotone pseudo-gradient. By Corollary \ref{corr::RDConvergence}, RD converges in $\Gamma^H$ which, from Lemma \ref*{lem::QLRD} gives convergence of Q-Learning.
    \end{proof}

    The proof of Lemma \ref*{lem::PotentialConv} relies on the following proposition.

    \begin{proposition}
        Let $g \, : \mathcal{X} \rightarrow \mathcal{Y}$ be an operator with derivative $Dg \, : \mathcal{X} \rightarrow \mathcal{Y}$. If $g$ is convex, then $Dg$ is monotone. 
    \end{proposition}
    \begin{proof}
        Suppose for the sake of contradiction that $Dg$ is not monotone. I.e. that, for some $\x, \x' \in \mathcal{X}$
        \begin{equation}
            \langle \x - \x', Dg(\x) - Dg(\x') \rangle < 0. 
        \end{equation}

        From the convexity of $g$, we have that
        \begin{align}
            g(\x) & \geq g(\x') + \langle \x - \x', Dg(\x') \rangle \\
            g(\x') & \geq g(\x) + \langle \x' - \x, Dg(\x) \rangle
        \end{align}

        Taking the sum, we have that
        \begin{align}
            g(\x) + g(\x') & \geq g(\x) + g(\x') - \langle \x - \x', Dg(\x) - Dg(\x') \rangle
        \end{align}

        which is a contradiction.
    \end{proof}

    \begin{lemma} \label{lem::ConcavePotential}
        Any weighted potential game $\Gamma$ with concave potential $
        (\x)$ has a monotone pseudo-gradient $F(\x)$.
    \end{lemma}

    \begin{proof}
        By the definition of the weighted potential game, there are positive 
        constants $w_1, \ldots, w_N > 0$ such that
        \begin{equation}
            D_{\x_k} U(\x) = w_k D_{\x_k} u_k(\x) = - w_k F_k(\x).
        \end{equation}

        Since $U(\x)$ is concave and $w_k > 0$, $ - D_{\x_k} u_k(\x)$ is monotone. Taking the sum 
        over all agents $k$, we achieve the required result.
    \end{proof}

    \begin{proof}[Proof of Lemma \ref*{lem::PotentialConv}]
        By Lemma \ref*{lem::ConcavePotential} we have that, if the potential is concave (resp.
        strictly concave), then $F(\x)$ is monotone (resp. strictly monotone). In the concave case,
        we have from Lemma \ref*{lem::QLMonotone} that the perturbed game has a strictly monotone
        $F(\x)$. Then, from Theorem \ref*{thm::QLConvergence} we have convergence of Q-Learning.
    \end{proof}

    \begin{proposition} \label{prop::NZSGProp}
        For any weighted polymatrix zero sum game $\Gamma$,
        \begin{equation}
            \sum_k w_k \langle \x_k', \sum_{(k, l \in \edgeset) }A^{kl}\x_l \rangle + \sum_k \langle \x_k, \sum_{(k, l \in \edgeset) }A^{kl}\x_l' \rangle = 0
        \end{equation}
        
        for any $\x, \x' \in \Delta$
    \end{proposition}
    \begin{proof}
        This proposition follows directly from \cite{kadan:exponential} (Lemma 1) which considers general payoffs in network zero sum games and is also an adjustment of \cite{piliouras:zerosum} (Lemma 4.3) which considers $\x'$ to specifically be the QRE $\NE$. For the sake of completeness, however, we reproduce the proof by \cite{kadan:exponential} here
        
        \begin{align}
            - w_k u_k(\y_k, \x_{-k}) &= \sum_{j \in \agentset_{-k}} w_j \left( u_{jk}(\x_j, \y_k) + \sum_{(j, l) \in \edgeset_{-k}} u_{jl}(\x_j, \x_l) \right) \nonumber \\
            - \sum_k w_k u_k(\y_k, \x_{-k}) &= \sum_k \sum_{j \in \agentset_{-k}} w_j \left( u_{jk}(\x_j, \y_k) + \sum_{(j, l) \in \edgeset_{-k}} u_{jl}(\x_j, \x_l) \right) \nonumber \\
             &= \sum_k \sum_{j \in \agentset_{-k}} w_j u_{jk}(\x_j, \y_k) + \sum_k \sum_{j \in \agentset_{-k}} \sum_{(j, l) \in \edgeset_{-k}} w_j u_{jl}(\x_j, \x_l) \nonumber \\
             &= \sum_k w_k u_{k}(\x_k, \y_{-k})
        \end{align}
        
        from which the result follows.
        
    \end{proof} 






    \begin{proof}[Proof of Lemma \ref*{lem::NZSGConv}]
        For a weighted zero-sum polymatrix game, for any $\x, \x' \in \Delta$
        \begin{align*}
            - \langle \x - \x', F(\x; \w) - F(\x'; \w) \rangle &= \sum_k w_k \langle \x_k - \x_k', \sum_{(k, l \in \edgeset) }A^{kl}\x_l - \sum_{(k, l \in \edgeset) }A^{kl}\x_l'  \rangle \\
            &= \sum_k w_k \langle \x_k', \sum_{(k, l \in \edgeset) }A^{kl}\x_l' \rangle + \sum_k w_k \langle \x_k, \sum_{(k, l \in \edgeset) }A^{kl}\x_l \rangle \\ & + \sum_k w_k \langle \x_k', \sum_{(k, l \in \edgeset) }A^{kl}\x_l \rangle +  \sum_k w_k \langle \x_k, \sum_{(k, l \in \edgeset) }A^{kl}\x_l' \rangle \\
            &= 0
        \end{align*}
        
        where the first two terms of the final equality are zero due to the weighted zero sum property and the final two are zero due to Prop. \ref{prop::NZSGProp}. Therefore, the game is weighted monotone. Convergence to a unique QRE holds due to Theorem \ref*{thm::QLConvergence}.
    \end{proof}

    \bibliographystyle{ieeetr}
    \bibliography{references}
    
\end{document}